\documentclass[english]{article}
\usepackage[T1]{fontenc}
\usepackage[latin9]{inputenc}
\usepackage{geometry}
\usepackage{color}
\usepackage{units}
\usepackage{amsmath}
\usepackage{amsthm}
\usepackage{amssymb}
\usepackage{graphicx}

\makeatletter
\theoremstyle{plain}
\newtheorem{thm}{\protect\theoremname}
  \theoremstyle{plain}
  \newtheorem{lem}[thm]{\protect\lemmaname}
  \theoremstyle{definition}
  \newtheorem{defn}[thm]{\protect\definitionname}
  \theoremstyle{definition}
  \newtheorem{example}[thm]{\protect\examplename}
\newenvironment{lyxlist}[1]
{\begin{list}{}
{\settowidth{\labelwidth}{#1}
 \setlength{\leftmargin}{\labelwidth}
 \addtolength{\leftmargin}{\labelsep}
 }}
{\end{list}}

\usepackage{fullpage}

\makeatother

\usepackage{babel}
  \providecommand{\definitionname}{Definition}
  \providecommand{\examplename}{Example}
  \providecommand{\lemmaname}{Lemma}
\providecommand{\theoremname}{Theorem}

\begin{document}

\title{Parity separation:\\
A scientifically proven method for permanent weight loss}

\author{Radu Curticapean\thanks{Simons Institute for the Theory of Computing, Berkeley, USA, and Institute for Computer Science and Control, Hungarian Academy of Sciences (MTA SZTAKI), Budapest, Hungary. Supported by ERC Grant 280152 PARAMTIGHT.}
\date{}}
\maketitle
\begin{abstract}

\global\long\def\FP{\mathsf{FP}}
\global\long\def\sharpP{\mathsf{\#P}}

\global\long\def\DP{\mathsf{P}}
\global\long\def\NP{\mathsf{NP}}
\global\long\def\coNP{\mathsf{coNP}}

\global\long\def\SETH{\mathsf{SETH}}
\global\long\def\ETH{\mathsf{ETH}}
\global\long\def\sharpETH{\mathsf{\#ETH}}
\global\long\def\sharpSETH{\mathsf{\#SETH}}
\global\long\def\rETH{\mathsf{rETH}}

\global\long\def\FPT{\mathsf{FPT}}
\global\long\def\Wone{\mathsf{W[1]}}
\global\long\def\sharpWone{\mathsf{\#W[1]}}
 \global\long\def\sharpWtwo{\mathsf{\#W[2]}}

\global\long\def\CeqP{\mathsf{C_{=}P}}
\global\long\def\GapP{\mathsf{GapP}}
\global\long\def\parityP{\mathsf{\oplus P}}
\global\long\def\PP{\mathsf{PP}}

\global\long\def\sharpSAT{\mathsf{\#SAT}}
\global\long\def\SAT{\mathsf{SAT}}
\global\long\def\colHS{\mathsf{ColHS}}

\global\long\def\HittingSet{\mathsf{HittingSet}}
\global\long\def\pClique{\mathsf{\#Clique}}
\global\long\def\pHittingSet{\mathsf{\#HittingSet}}
\global\long\def\A{\mathsf{A}}
\global\long\def\B{\mathsf{B}}

\global\long\def\perm{\mathrm{perm}}
\global\long\def\usat{\mathrm{usat}}
\global\long\def\PM{\mathcal{PM}}
\global\long\def\DM{\mathcal{DM}}
\global\long\def\M{\mathcal{M}}
\global\long\def\PerfMatch{\mathrm{PerfMatch}}
\global\long\def\MatchSum{\mathrm{MatchSum}}

\global\long\def\Holant{\mathrm{Holant}}
\global\long\def\Sig{\mathrm{Sig}}
\global\long\def\supp{\mathrm{supp}}
\global\long\def\hw{\mathrm{hw}}
\global\long\def\val{\mathrm{val}}
\global\long\def\one{\mathtt{1}}
\global\long\def\zero{\mathtt{0}}

\global\long\def\cw{\mathrm{cw}}
\global\long\def\pw{\mathrm{pw}}
\global\long\def\tw{\mathrm{tw}}
\global\long\def\isol{\mathrm{isol}}
\global\long\def\crossnum{\mathrm{cr}}

\global\long\def\leqPoly{\leq_{\mathit{p}}}
\global\long\def\leqPolyT{\leq_{\mathit{p}}^{T}}
\global\long\def\leqPolyPar{\leq_{p}^{\mathit{pars}}}

\global\long\def\leqFpt{\leq_{\mathit{fpt}}}
\global\long\def\leqFptT{\leq_{\mathit{fpt}}^{T}}
\global\long\def\leqFptPar{\leq_{\mathit{fpt}}^{\mathit{pars}}}
\global\long\def\leqSERF{\leq_{\mathit{serf}}}
\global\long\def\leqPLin{\leq_{p}^{\mathit{lin}}}
\global\long\def\leqFptLin{\leq_{\mathit{fpt}}^{\mathit{lin}}}

\global\long\def\Matchings{\mathcal{M}}
\global\long\def\UCWs{\mathcal{U}}
\global\long\def\Cycles{\mathcal{C}}
\global\long\def\PathCycles{\mathcal{PC}}
\global\long\def\modulo{\,\mathrm{mod}\,}
\global\long\def\vec#1{\mathbf{#1}}
\global\long\def\H{\mathcal{H}}

\global\long\def\EQ{\mathtt{EQ}}
\global\long\def\sigHW#1{\mathtt{HW_{#1}}}
\global\long\def\sigEven{\mathtt{EVEN}}
\global\long\def\sigOdd{\mathtt{ODD}}

\global\long\def\todo#1{\mbox{{\bf TODO:}\,#1}\ }

Given an edge-weighted graph $G$, let $\PerfMatch(G)$
be the following weighted sum that ranges over all perfect matchings
$M$ in $G$:

\[
\PerfMatch(G):=\sum_{M}\prod_{e\in M}w(e).
\]
If $G$ is unweighted, this plainly counts the perfect matchings of
$G$.

In this paper, we introduce\emph{ parity separation}, a new method
for reducing $\PerfMatch$ to unweighted instances: For graphs $G$
with edge-weights $\pm1$, we construct two unweighted graphs $G_{1}$
and $G_{2}$ such that 
\[
\PerfMatch(G)=\PerfMatch(G_{1})-\PerfMatch(G_{2}).
\]
This yields a novel weight removal technique for counting perfect
matchings, in addition to those known from classical $\sharpP$-hardness
proofs. We derive the following applications:
\begin{enumerate}
\item An alternative $\sharpP$-completeness proof for counting unweighted
perfect matchings.{\small \par}
\item $\CeqP$-completeness for deciding whether two given unweighted graphs
have the same number of perfect matchings. To the best of our knowledge,
this is the first $\CeqP$-completeness result for the ``equality-testing
version'' of any natural counting problem that is not already $\sharpP$-hard
under parsimonious reductions.{\small \par}
\item An alternative tight lower bound for counting unweighted perfect matchings
under the counting exponential-time hypothesis $\sharpETH$.{\small \par}
\end{enumerate}
Our technique is based upon matchgates and the Holant framework. To
make our $\sharpP$-hardness proof self-contained, we also apply matchgates
for an alternative $\sharpP$-hardness proof of $\PerfMatch$ on graphs
with edge-weights $\pm1$.
\end{abstract}
\clearpage{}

\section{Introduction}

The problem of counting perfect matchings has played a central role
in counting complexity since Valiant \cite{Valiant1979a} introduced
the class $\sharpP$ and established $\sharpP$-completeness of counting
perfect matchings in unweighted bipartite graphs. This problem was
previously already considered in statistical physics \cite{Temperley.Fisher1961,Kasteleyn1961,Kasteleyn1967},
and Valiant's computational hardness result explains the lack of progress
encountered in this area for finding efficient algorithms for counting
perfect matchings.

As complexity theorists, we can appreciate this seminal $\sharpP$-completeness
result from another perspective: The problem of counting perfect matchings
in \emph{unweighted} graphs presented the first example of a natural
hard counting problem with an easy decision version, since Edmond's
classical algorithm \cite{Edmonds1987} allows to decide in polynomial
time whether a graph contains \emph{at least one} perfect matching.
This showed exemplarily that the complexity-theoretic study of counting
problems amounts to more than merely checking whether $\NP$-hardness
proofs for decision problems carry over to their counting versions.

For instance, a fundamental peculiarity of counting problems that
is not shared by decision problems are \emph{cancellations}: In (weighted)
counting problems, witness structures may cancel each other out, and
this can have strong effects on the complexity of the problem. The
most prominent example of this phenomenon might be the situation of
the \emph{determinant} and the \emph{permanent}, both summing over
the same permutations, however with different weights. This results
in the permanent being $\sharpP$-complete by Valiant's result, whereas
the determinant can be computed in polynomial time. The \emph{accidental
}and \emph{holographic }algorithms introduced by Valiant \cite{Valiant2006,Valiant2008}
provide examples for further and more unexpected cancellations that
render counting problems easy.

However, cancellations are also crucial for negative results: In many
$\sharpP$-hardness proofs, such as \cite{Blaeser.Dell2007,Blaeser.Curtican2011,Blaeser.Curtican2012},
we first define an intermediate variant of the target problem on weights
$\pm1$. Examples for this strategy include the original reduction
from $\sharpSAT$ to counting unweighted perfect matchings \cite{Valiant1979a}:
In this setting, let $G$ be a graph with edge-weights $w:E(G)\to\{-1,1\}$,
let $\PM[G]$ denote its set of perfect matchings, and define 
\begin{equation}
\PerfMatch(G):=\sum_{M\in\PM[G]}\prod_{e\in M}w(e).\label{eq: Def PerfMatch}
\end{equation}

Given an instance to this weighted problem, that is, a graph $G$
derived from a $3$-CNF formula, its space of witness structures $\PM[G]$
can then be partitioned into ``good'' structures that correspond
to satisfying assignments, and ``bad'' structures that could be
called combinatorial noise. By careful construction of such a graph
$G$ on edge-weights $\pm1$, we can ensure that bad structures come
in pairs of weight $+1$ and $-1$, thus canceling out, whereas good
structures all have weight $+1$.

To conclude $\sharpP$-completeness of counting unweighted perfect
matchings, it remains to simulate the weight $-1$ from the intermediate
problem. This can be achieved by several techniques, which we survey
soon. Let us however first point out that the main contribution of
this paper is a novel technique for precisely this part of the reduction:
Using a method we call \emph{parity separation}, we reduce the computation
of $\PerfMatch(G)$ for a $\pm1$-weighted graph $G$ to the difference
of $\PerfMatch$ for two unweighted graphs, that is, to the difference
of two numbers of perfect matchings.
\begin{lem}[Parity Separation]
\label{main lem: Difference technique}Let $G$ be a graph on $n$
vertices and $m$ edges that is weighted by a function $w:E(G)\to\{-1,1\}$.
Then we can construct in time $\mathcal{O}(n+m)$ two unweighted graphs
$G_{1}$ and $G_{2}$, each on $\mathcal{O}(n+m)$ vertices and edges,
such that 
\begin{equation}
\PerfMatch(G)=\PerfMatch(G_{1})-\PerfMatch(G_{2}).\label{eq: Difference}
\end{equation}

\end{lem}
Intuitively speaking, this allows us to ``collect'' positive and
negative terms of $\PerfMatch(G)$ for $\pm1$-weighted graphs. This
way, we can reduce the effect of cancellations incurred \emph{within}
$\PerfMatch$ to a mere difference \emph{outside} of $\PerfMatch$.

In the remainder of this introduction, we present parity separation
in more detail and demonstrate three applications that can be derived
from it: Firstly, and not surprisingly, we obtain a new $\sharpP$-completeness
proof for counting perfect matchings. Secondly, we can show $\CeqP$-completeness
of deciding whether two graphs have the same number of perfect matchings.
Thirdly, we also obtain tight lower bounds under the exponential-time
hypothesis.

\subsection{\label{sub: intro sharpP}$\protect\sharpP$-completeness via parity
separation}

To put parity separation into context, we first recapitulate Valiant's
$\sharpP$-hardness result for counting perfect matchings in more
detail. Let us denote the problem of evaluating $\PerfMatch$ on graphs
with edge-weights from $A\subseteq\mathbb{Q}$ by $\PerfMatch^{A}$.
For consistency with \cite{Dell.Husfeldt2014}, we include $0\in A$.

\subsubsection*{First step: From $\protect\sharpSAT$ to $\protect\PerfMatch^{-1,0,1}$}

It is shown in \cite[Lemma 3.1]{Valiant1979a} that $\PerfMatch^{W}$
is $\sharpP$-hard for $W:=\{-1,0,1,2,3\}$. More precisely, from
a $3$-CNF formula $\varphi$, a number $t(\varphi)\in\mathbb{N}$
and a bipartite graph $G=G(\varphi)$ on weights $W$ are constructed
in polynomial time, such that 
\begin{equation}
\sharpSAT(\varphi)=\frac{\PerfMatch(G)}{4^{t(\varphi)}},\label{eq: SAT to PerfMatch(1,-1)}
\end{equation}

This however only yields hardness for a weighted generalization of
counting perfect matchings. To obtain a useful reduction source for
further problems, it is crucial to reduce $\PerfMatch^{W}$ to $\PerfMatch^{0,1}$,
as reductions from $\PerfMatch$ to other problems would otherwise
need to take care of the weights in $W$, which is particularly problematic
for the edge-weight $-1$ in the case of unweighted reduction targets. 

In fact, the weight $-1$ is the only problem we encounter: Edges
$e$ of positive integer edge-weight $w$ can be simulated easily
by replacing $e$ with $w$ parallel edges of unit weight, possibly
subdividing edges twice to obtain simple graphs. This trick however
does not apply for the weight $-1$, so we need a different strategy.

\subsubsection*{Second step: Removing the edge-weight $-1$}

By now, two different strategies are known for this step, which we
briefly survey in the following. Let $G$ be a graph with $n$ vertices
and $m>0$ edges, all on weights $-1$ and $1$.
\begin{enumerate}
\item \textbf{Modular arithmetic:} Essentially the following approach was
originally used by Valiant \cite{Valiant1979a} and later refined
by Zankó \cite{DBLP:journals/ijfcs/Zanko91}: Write $M=2^{m}+1$ and
observe that $\PerfMatch(G)<M$. We can hence replace the weight $-1$
by the positive integer $M-1$ to obtain a graph $G'$ satisfying
$\PerfMatch(G)\equiv\PerfMatch(G')$ modulo $M$. The weight $M-1$
can be simulated by a gadget as in the previous paragraph, and using
a more involved construction \cite{DBLP:journals/ijfcs/Zanko91},
it can be seen that a gadget on $\mathcal{O}(m)$ vertices and edges
suffices, yielding a total number of $\mathcal{O}(nm)$ vertices and
$\mathcal{O}(m^{2})$ edges in $G'$. Then we compute $\PerfMatch(G')$
modulo $M$ and obtain $\PerfMatch(G)$, as we may assume from (\ref{eq: SAT to PerfMatch(1,-1)})
that $\PerfMatch(G)\geq0$. In total, we obtain one reduction image
for $\PerfMatch^{0,1}$ on $\mathcal{O}(nm)$ vertices and $\mathcal{O}(m^{2})$
edges.
\item \textbf{Polynomial interpolation:} An alternative technique for removing
the edge-weight $-1$ from $G$ is to replace it by an indeterminate
$x$. This gives rise to a graph $G_{x}$ on edge-weights $\{1,x\}$
for which $\PerfMatch(G_{x})$ is a polynomial $p(x)\in\mathbb{Z}[x]$
of degree at most $n/2$. We can evaluate $p(i)$ for $i\in\{0,\ldots,n/2\}$
by substituting $x\gets i$ in $G_{x}$ and simulating this positive
weight by a gadget as discussed before. This allows us to recover
$p(-1)=\PerfMatch(G)$ via Lagrangian interpolation. In total, using
gadgets as in \cite{DBLP:journals/ijfcs/Zanko91,Dell.Husfeldt2014},
we obtain $\mathcal{O}(n)$ reduction images for $\PerfMatch^{0,1}$
on $\mathcal{O}(n\log m)$ vertices and $\mathcal{O}(n\log m+m)$
edges each.
\end{enumerate}
Both weight removal techniques allow to reduce $\PerfMatch^{-1,0,1}$
to $\PerfMatch^{0,1}$ and thus complete the $\sharpP$-completeness
proof of the latter problem. Note however that both approaches map
weighted graphs $G$ with $m$ edges to unweighted graphs with a superlinear
number of edges. Using parity separation, we obtain a third way of
performing the weight removal step, which differs substantially from
both approaches mentioned before and features only constant blowup:
\begin{enumerate}
\item[3.] \textbf{Parity separation:} Using Lemma~\ref{main lem: Difference technique},
compute two unweighted graphs $G_{1}$ and $G_{2}$ from $G$ such
that $\PerfMatch(G)$ is the mere difference of $\PerfMatch(G_{1})$
and $\PerfMatch(G_{2})$. In total, we obtain $2$ reduction images
for $\PerfMatch^{0,1}$ on $\mathcal{O}(n+m)$ vertices and edges.
\end{enumerate}
Together with the first step, this implies an alternative $\sharpP$-completeness
proof for $\PerfMatch^{0,1}$. We note that, for sake of completeness,
we will later also give a self-contained proof of the first step.
\begin{thm}
\label{thm: sharpP-completeness}The problem $\PerfMatch^{0,1}$ is
$\sharpP$-complete under polynomial-time Turing reductions.
\end{thm}

\subsection{\label{sub: intro CeqP}$\protect\CeqP$-completeness via parity
separation}

Apart from an alternative $\sharpP$-completeness proof, Lemma~\ref{main lem: Difference technique}
also yields implications for the structural complexity of $\PerfMatch$:
We show that deciding whether two unweighted graphs have the same
number of perfect matchings is complete for the complexity class $\CeqP$
introduced in \cite{SimonThesis} and elaborated in \cite{DBLP:journals/sigact/HemaspaandraV95,Hemaspaandra.Ogihar2002}.

To define $\CeqP$, let us associate the following language $\A_{=}$
with each counting problem $\A\in\sharpP$: The inputs to $\A_{=}$
are pairs $(x,y)$ of instances to $\A$, and we are asked to determine
whether $\A(x)=\A(y)$ holds. We can then define\footnote{We deviate here from the standard definition of $\CeqP$, according
to which we have $L\in\CeqP$ iff there is a polynomial-time nondeterministic
Turing machine $M$ such that $x\in L$ iff the numbers of accepting
and rejecting computation paths of $M(x)$ are equal. It can be verified
easily that this is equivalent to our definition.} the class 
\[
\CeqP:=\{\A_{=}\mid\A\in\sharpP\}.
\]

For instance, it is clear that $\sharpSAT_{=}$, the problem that
asks whether two 3-CNF formulas have the same number of satisfying
assignments, is $\CeqP$-complete under polynomial-time many-one reductions.
In fact, $\CeqP$-completeness holds for every problem $\mathsf{A}_{=}$
whose counting version $\mathsf{\#A}$ is $\sharpP$-complete under
parsimonious reductions. We recall the notion of parsimonious (and
other) reductions in Definition~\ref{def: Reductions}.

The relationship between $\CeqP$ and other complexity classes has
been studied in structural complexity theory, and several results
are surveyed in \cite{Hemaspaandra.Ogihar2002}. For instance, we
clearly have $\mathsf{coNP}\subseteq\CeqP$, and using the witness
isolation technique \cite{Valiant.Vazirani1986}, we see that $\NP$
is contained in $\CeqP$ under randomized reductions. Let us also
observe that $\NP^{\sharpP}\subseteq\NP^{\CeqP}$: Whenever we issue
an oracle call to $\sharpP$, we may instead guess the output number,
and then check whether we guessed correctly by using the $\CeqP$
oracle. 

To the best of the author's knowledge, no \emph{natural} $\CeqP$-complete
problem $\A_{=}$ is known whose counting version $\A$ is \emph{not}
$\sharpP$-complete under parsimonious reductions.\footnote{Here, we stressed \emph{natural}, because we can easily construct
artificial $\CeqP$-complete problems $\mathsf{A_{=}}$ whose counting
version $\mathsf{\#A}$ admits no parsimonious reduction from $\sharpSAT$:
Consider as an example the counting problem $\mathsf{\#SAT'}$ that
asks to count satisfying assignments, incremented by $1$. If $\mathsf{\#SAT'}$
had a parsimonious reduction from $\sharpSAT$, then every CNF-formula
would be satisfiable. On the other hand, the reduction from $\sharpSAT_{=}$
to $\mathsf{\#SAT'_{=}}$ is trivial. } It is clear that the problem $\PerfMatch^{0,1}$ of counting unweighted
perfect matchings cannot be $\sharpP$-complete under parsimonious
reductions, unless $\mathsf{P}=\NP$. Therefore, the following completeness
result for $\PerfMatch_{=}^{0,1}$ seems relevant for structural complexity
theory, as it establishes a $\CeqP$-variant of Valiant's result.
\begin{thm}
\label{main thm: CeqP}The problem $\PerfMatch_{=}^{0,1}$ is $\CeqP$-complete
under polynomial-time many-one reductions: Decide whether unweighted
graphs $G_{1}$ and $G_{2}$ have the same number of perfect matchings.
\end{thm}
To prove this theorem, we first reduce instances $(\varphi,\varphi')$
for $\sharpSAT_{=}$ to $\pm1$-weighted graphs $G$ that satisfy
$\PerfMatch(G)=0$ iff $\sharpSAT(\varphi)=\sharpSAT(\varphi')$.
This requires a modification of the first step in the $\sharpP$-hardness
reduction, which is however supported easily by our alternative proof.
Then we apply Lemma~\ref{main lem: Difference technique} on the
graph $G$ to obtain unweighted graphs $G_{1}$ and $G_{2}$ satisfying
(\ref{eq: Difference}). In particular, their numbers of perfect matchings
agree iff $\PerfMatch(G)$ vanishes, that is, iff $(\varphi,\varphi')$
is a yes-instance for $\sharpSAT_{=}$.

To conclude this subsection, we note that the complexity of a similar
problem was posed as an open question in \cite{Chen2010}: Given two
directed acyclic graphs, decide whether their numbers of topological
orderings agree. It was shown in \cite{Brightwell.Winkler1991} that
counting topological orderings is $\sharpP$-complete under Turing
reductions, but the decision version is trivial for acyclic graphs.
Our result for $\PerfMatch_{=}^{0,1}$ might be useful to prove $\CeqP$-completeness
for this and other problems.

\subsection{\label{sub: intro ETH}Tight lower bounds via parity separation}

We turn our attention to conditional \emph{quantitative} lower bounds:
It is a recent trend in computational complexity to make use of assumptions
stronger than $\DP\neq\NP$ or $\FP\neq\sharpP$ to prove tight (exponential)
lower bounds on the running times needed to solve computational problems.
A popular such assumption is the exponential-time hypothesis $\ETH$,
introduced by Impagliazzo et al.~\cite{Impagliazzo.Paturi2001,Impagliazzo.Paturi2001a},
which states that the satisfiability of $n$-variable formulas $\varphi$
in $3$-CNF cannot be decided in time $2^{o(n)}$. For counting problems,
an analogous variant $\sharpETH$ was introduced by Dell et al.~\cite{Dell.Husfeldt2014},
and it postulates the same for the problem of \emph{counting} satisfying
assignments to $\varphi$.

Assuming $\ETH$, it was shown for a vast body of popular decision
problems that the known exponential-time exact algorithms are somewhat
optimal: For instance, there is a trivial $2^{\mathcal{O}(m)}$ time
algorithm for finding a Hamiltonian cycle (or various other structures)
in an $m$-edge graph, but $2^{o(m)}$ time algorithms would refute
$\ETH$. See \cite{Lokshtanov.Marx2011a} for an accessible survey.

Similar lower bounds were shown for counting problems under $\sharpETH$,
see \cite{Hoffmann2010,Husfeldt.Taslam2010,Dell.Husfeldt2014}, and
a very recent paper \cite{DBLP:conf/icalp/Curticapean15} introduced
\emph{block interpolation}, an approach to make the technique of polynomial
interpolation (as seen in the second step of Section~\ref{sub: intro sharpP})
compatible with tight lower bounds under $\sharpETH$. For several
problems, that of counting perfect matchings being among them, block
interpolation gave the first tight $2^{\Omega(m)}$ lower bounds under
$\sharpETH$. 

When applying this framework to $\PerfMatch^{0,1}$, we would first
reduce $\sharpSAT$ on $n$-variable $3$-CNFs $\varphi$ to instances
$G=G(\varphi)$ for $\PerfMatch^{-1,0,1}$ with $\mathcal{O}(n)$
edges as in the first step of the $\sharpP$-hardness proof. Then
we apply the block interpolation technique to reduce $G$ to $2^{o(n)}$
unweighted instances $G'$ for $\PerfMatch^{0,1}$ with $\mathcal{O}(n)$
edges. While this sub-exponential number of instances is compatible
with the goal of proving tight lower bounds, it leaves open the natural
question whether the same reduction could be achieved with only polynomially
many oracle calls on graphs with $\mathcal{O}(n)$ edges.

Using Lemma~\ref{main lem: Difference technique}, we obtain a strong
positive answer to this question: Replacing the application of block
interpolation by one of parity separation, we obtain a reduction to
merely \emph{two} instances of $\PerfMatch^{0,1}$. And as a synthesis
of structural and quantitative complexity, we also obtain a tight
lower bound for the equality-testing problem $\PerfMatch_{=}^{0,1}$.
\begin{thm}
\label{main thm: Bounds under ETH}Unless $\sharpETH$ fails, the
problem $\PerfMatch^{0,1}$ cannot be solved in time $2^{o(m)}$ on
simple graphs with $m$ edges. The same applies to $\PerfMatch_{=}^{0,1}$
under the decision version $\ETH$.
\end{thm}

\subsection*{Organization of this paper}

The remainder of this paper is structured as follows: In Section~\ref{sec: Preliminaries},
we introduce the Holant framework and matchgates, concepts that are
crucial to our constructions. These are put to use in Section~\ref{sec: Parity Separation},
where we prove Lemma~\ref{main lem: Difference technique}, our main
result. Its applications, as discussed above, are shown in Section~\ref{sec: Applications}.
Omitted proofs can be found in the Appendix.

\section{\label{sec: Preliminaries}Preliminaries}

Graphs in this paper may be edge- or vertex-weighted. Given a graph
$G$ and $v\in V(G)$, denote the edges incident with $v$ by $I(v)$.
If the context of an argument unambiguously determines a graph $G$,
we write $n=|V(G)|$ and $m=|E(G)|$.

We denote the Hamming weight of strings $x\in\{0,1\}^{*}$ by $\hw(x)$.
Given a statement $\varphi$, we let $[\varphi]=1$ if $\varphi$
is true, and $[\varphi]=0$ otherwise. For convenience, we recall
that several reduction notions are distinguished in the study of counting
complexity: The most restrictive notion is that of \emph{parsimonious}
(many-one) reductions, which can be slightly relaxed to \emph{weakly}
parsimonious reductions. The most permissive notion is that of \emph{Turing
}reductions.
\begin{defn}
\label{def: Reductions}Let $\A$ and $\B$ be counting problems.
Let $f:\{0,1\}^{*}\to\{0,1\}^{*}$ and $g:\{0,1\}^{*}\to\mathbb{Q}$
be polynomial-time computable functions. If $\A(x)=g(x)\cdot\B(f(x))$
holds for all $x\in\{0,1\}^{*}$, then we call $(f,g)$ a \emph{weakly
parsimonious (polynomial-time) reduction} from $\A$ to $\B$ and
write $\A\leqPoly\B$. If additionally $g(x)=1$ holds for all $x\in\{0,1\}^{*}$,
then we call $f$ \emph{parsimonious} and write $\A\leqPolyPar\B$.

If $\mathbb{T}$ is a deterministic polynomial-time algorithm that
solves $\A$ with an oracle for $\B$, then we call $\mathbb{T}$
a \emph{Turing reduction} from $\A$ to $\B$ and write $\A\leqPolyT\B$.
\end{defn}

\subsection{Weighted sums of (perfect) matchings}

The quantity $\PerfMatch$ on edge-weighted graphs, as defined in
(\ref{eq: Def PerfMatch}) and \cite{Valiant2008}, will be the central
object of investigation in this paper. For intermediate steps, we
also consider the quantity $\MatchSum$ from \cite{Valiant2008}.
\begin{defn}
For vertex-weighted graphs $G$ with $w:V(G)\to\mathbb{Q}$, let $\M[G]$
denote the set of (not necessarily perfect) matchings in $G$. Recall
that $\PM[G]\subseteq\M[G]$ denotes the perfect matchings in $G$.
For $M\in\M[G]$, let $\usat(M)$ denote the set of unmatched vertices
in $M$. Then we define 
\[
\MatchSum(G)=\sum_{M\in\M[G]}\prod_{v\in\usat(M)}w(v).
\]

\end{defn}
Given $W\subseteq\mathbb{Q}$, we write $\PerfMatch^{W}$ for the
problem of evaluating $\PerfMatch(G)$ on graphs $G$ with weights
$w:E(G)\to W$. Likewise, write $\MatchSum^{W}$ on graphs with weights
$w:V(G)\to W$. Please note that an edge of weight $0$ in $\PerfMatch$
can be treated as if it were not present, whereas weight $0$ at a
vertex $v$ in $\MatchSum$ signifies that $v$ must be matched.

We can easily reduce $\PerfMatch^{W}$ for finite $W\subseteq\mathbb{Q}$
to $\PerfMatch^{-1,0,1}$:
\begin{lem}[folklore]
\label{lem: Fractional weights in PerfMatch} Let $G$ be edge-weighted
by $w:E(G)\to\mathbb{Q}$. Let $q\in\mathbb{N}$ denote the lcd of
the weights in $G$, and let $T=q\cdot\max_{e\in E(G)}w(e)$. Then
we can compute a number $B\in\mathbb{N}$ and an edge-weighted graph
$G'$ on $\mathcal{O}(n+Tm)$ vertices and edges, all of weight $\pm1$,
such that $\PerfMatch(G)=q^{-B}\cdot\PerfMatch(G')$.\end{lem}
\begin{proof}
Define a graph $G_{1}$ from $G$ by declaring $w_{1}(e)=q\cdot w(e)$
for $e\in E(G)$. Then
\[
\PerfMatch(G)\ =\ \sum_{M\in\PM[G]}\prod_{e\in M}w(e)\ =\ \sum_{M\in\PM[G_{1}]}\prod_{e\in M}\frac{w_{1}(e)}{q}\ =\ q^{-n/2}\cdot\PerfMatch(G_{1}).
\]
We construct a graph $G_{2}$ from $G_{1}$ in which the only negative
edge-weight appearing is $-1$: If $e=uv$ is an edge with negative
weight $w(e)\in\mathbb{Z}$, we can subdivide $e$ twice to obtain
subdivision vertices $s_{1}$ and $s_{2}$. Then assign weight $|w(e)|$
to the edge $us_{1}$, assign weight $-1$ to the edge $s_{2}v$,
and weight $1$ to the edge $s_{1}s_{2}$.

Finally, we obtain $G'$ from $G_{2}$ by simulating each edge-weight
$w>0$ of  $G_{2}$ by $w$ parallel edges, subdivided twice to obtain
a simple graph, as previously mentioned in the introduction.
\end{proof}

\subsection{Holant problems}

We give an introduction to the \emph{Holant framework}, summarizing
ideas from \cite{Valiant2008,Cai.Lu2007,Cai.Lu2008}. A more detailed
introduction to the notation used in this subsection and the following
one can be found in \cite{Curticapean.PhD}.\global\long\def\assignment{x\in\{0,1\}^{E(\Omega)}}

\begin{defn}[adapted from~\cite{Valiant2008}]
\label{def: holant} A \emph{signature graph} is an edge-weighted
graph $\Omega$, which may feature parallel edges, with a \emph{vertex
function} $f_{v}:\{0,1\}^{I(v)}\to\mathbb{Q}$ at each $v\in V(\Omega)$.

The \emph{Holant} of $\Omega$ is a particular sum over edge assignments
$\assignment$. We sometimes identify $x$ with $x^{-1}(1)$. Given
$S\subseteq E(\Omega)$, we write $x|_{S}$ for the restriction of
$x$ to $S$, which is the unique assignment in $\{0,1\}^{S}$ that
agrees with $x$ on $S$. Then we define
\begin{equation}
\Holant(\Omega):=\sum_{\assignment}\left(\prod_{e\in x}w(e)\right)\left(\prod_{v\in V(\Omega)}f_{v}(x|_{I(v)})\right).\label{eq: Holant}
\end{equation}

\end{defn}
As a first example, we can reformulate $\PerfMatch(G)$ easily as
the Holant problem of a signature graph $\Omega=\Omega(G)$ by declaring
$f_{v}:\{0,1\}^{I(v)}\to\{0,1\}$ for $v\in V(G)$ to be the vertex
function that maps $x\in\{0,1\}^{*}$ to $1$ iff $\hw(x)=1$ and
to $0$ else.

When considering signature graphs $\Omega$ in the following, we will
always assume that $I(v)$ for each $v\in V(\Omega)$ is ordered in
a fixed (usually implicit) way. This way, if $v$ is a vertex of degree
$d\in\mathbb{N}$, we can view $f_{v}$ as a function $f_{v}:\{0,1\}^{d}\to\mathbb{Q}$,
and we call this representation a \emph{signature}.
\begin{example}
\label{exa: signatures}We consider signatures of arity $k\in\mathbb{N}$
on inputs $x\in\{0,1\}^{[k]}$ with $x=(x_{1},\ldots,x_{k})$. 
\begin{eqnarray*}
\EQ & : & x\mapsto[x_{1}=\ldots=x_{k}]\\
\mathtt{HW_{=1}} & : & x\mapsto[\hw(x)=1]\\
\mathtt{HW_{\leq1}} & : & x\mapsto[\hw(x)\leq1]\\
\mathtt{ODD} & : & x\mapsto x_{1}\oplus\ldots\oplus x_{k}\\
\mathtt{EVEN} & : & x\mapsto1\oplus x_{1}\oplus\ldots\oplus x_{k}.
\end{eqnarray*}
We may write, say, $\EQ_{4}$ to denote the arity-4 signature $\EQ$.
Note that these signatures are symmetric, as they depend only upon
the Hamming weight on the input.
\end{example}
Similarly as for $\PerfMatch$, we can also express $\MatchSum$ as
a Holant problem.
\begin{lem}
\label{lem: Holant MatchSum}Let $G$ be a graph with $w:V(G)\to\mathbb{Q}$.
Then $\MatchSum(G)=\Holant(\Omega)$ with the signature graph $\Omega$
derived from $G$ by placing $\mathtt{VTX}_{w(v)}$ at $v\in V(G)$.
Here, $\mathtt{VTX}_{w}$ for $w\in\mathbb{\mathbb{Q}}$ is 
\[
\mathtt{VTX}_{w}:\quad x\mapsto\begin{cases}
w & \mbox{if }\hw(x)=0,\\
1 & \mbox{if }\hw(x)=1,\\
0 & \mbox{otherwise}.
\end{cases}
\]
\end{lem}
\begin{proof}
In every satisfying assignment $\assignment$, each vertex $v\in V(\Omega)$
is incident with at most one active edge, so $x$ is a (not necessarily
perfect) matching, and $\val_{\Omega}(x)$ is the product of the following
factors: 
\begin{itemize}
\item If $v$ is incident with exactly \emph{one} active edge, then $v$
contributes $1$ to $\val_{\Omega}(x)$. 
\item If $v$ is incident with \emph{no} active edges, so $v\in\usat(G,x)$,
then $v$ contributes $w(v)$.
\end{itemize}
Hence, it holds that $\val_{\Omega}(x)=\prod_{v\in\usat(G,x)}w(v)$.
Summing over all matchings, we obtain $\Holant(\Omega)=\MatchSum(G)$
by identifying terms.
\end{proof}
We can easily reduce edge-weighted Holant problems to unweighted versions
as follows.
\begin{lem}
\label{lem: Removing Holant weights}Let $\Omega'$ be defined as
follows from $\Omega$: Subdivide each $e\in E(\Omega)$, assign weight
$1$ to the obtained subdivision edges, and equip the obtained subdivision
vertices with the signature $\mathtt{EDGE}_{w(e)}$, where 
\[
\mathtt{EDGE}_{w}:\quad x\mapsto\begin{cases}
w & \mbox{if }x=11,\\
0 & \mbox{if }x\in\{01,10\},\\
1 & \mbox{if }x=00.
\end{cases}
\]
Then $\Omega'$ features only the edge-weight $1$, and we have $\Holant(\Omega)=\Holant(\Omega')$.\end{lem}
\begin{proof}
The satisfying assignments $\assignment$ stand in bijection with
those of $\Omega'$: Every such $x$ can be transformed to a satisfying
assignment $x'\in\{0,1\}^{E(\Omega')}$ by assigning, for each $e\in E(\Omega)$,
the value $x(e)$ to both edges $e_{1},e_{2}$ obtained in $\Omega'$
from subdividing $e$.

Likewise, every such $x'\in\{0,1\}^{E(\Omega')}$ can be ``contracted''
to a unique satisfying assignment $\assignment$, since $x'(e_{1})=x'(e_{2})$
holds for every edge pair $e_{1,}e_{2}$ replacing an original edge
$e\in E(\Omega)$. We observe that $w_{\Omega}(x)\cdot\val_{\Omega}(x)=w_{\Omega'}(x')\cdot\val_{\Omega'}(x')$
holds, which shows the claim. 
\end{proof}
Finally, a signature is called \emph{even }if its support contains
only bitstrings of even Hamming weight. The problem $\sharpSAT$ can
be rephrased as a Holant problem with even signatures:
\begin{lem}
\label{lem: Holant SAT}For $n,m,d\in\mathbb{N}$, let $\varphi$
be a $d$-CNF formula on variables $x_{1},\ldots,x_{n}$ and clauses
$c_{1},\ldots,c_{m}$. We construct a signature graph $\Omega$ as
follows:
\begin{itemize}
\item For each $i\in[n]$, let $r(i)$ denote the number of occurrences
of $x_{i}$ (as a positive or negative literal) in $\varphi$. Create
a \emph{variable vertex} $v_{i}$ in $\Omega$, with signature $\EQ_{2r(i)}$.
\item For each $j\in[m]$, let $x_{i_{1}},\ldots,x_{i_{d}}$ be the variables
that clause $c_{j}$ depends upon. We create a \emph{clause vertex}
$w_{j}$ in $\Omega$, and for $\kappa\in[d]$, we add two parallel
edges between $w_{j}$ and $x_{i_{\kappa}}$ as the $2\kappa-1$-th
and $2\kappa$-th edges in the ordering of $I(w_{j})$. 
\item For each $j\in[m]$, consider clause $c_{j}$ as a Boolean function
on variables $x_{1},\ldots,x_{d}$. Define a function $c_{j}'$ on
variables $x_{1},\ldots,x_{2d}$ that outputs $c_{j}(x_{1},x_{3},\ldots,x_{2d-1})$
if $x_{2i}=x_{2i-1}$ for all $i\in[d]$. On all other inputs, the
value of $c_{j}'$ may be arbitrary. Assign such a signature $c_{j}'$
to the vertex $w_{j}$.
\end{itemize}
Then $\#\SAT(\varphi)=\Holant(\Omega)$ and $\Omega$ has $n+m$ vertices
and $2dm$ edges and only even signatures.\end{lem}
\begin{proof}
By the $\EQ$ signatures at variable vertices, every satisfying assignment
$\assignment$ corresponds to a unique binary assignment $x':\{x_{1},\ldots,x_{n}\}\to\{0,1\}$
to the variables of $\varphi$. Furthermore, for all $j\in[m]$, the
signature $c_{j}$ at the clause vertex $w_{j}$ ensures that $x'$
satisfies clause $c_{j}$, so altogether $x'$ satisfies $\varphi$.
Likewise, every satisfying assignment to $\varphi$ induces such a
satisfying assignment $\assignment$, thus proving the lemma. Note
that edges between variable and clause vertices come in pairs to ensure
that clauses vertices feature even signatures.
\end{proof}

\subsection{Gates and matchgates}

Given a signature graph $\Omega$, we can sometimes simulate vertex
functions by gadgets or \emph{gates}, which are signature graphs with
so-called \emph{dangling edges} that feature only one endpoint. These
notions are borrowed from the $\mathcal{F}$-gates in \cite{Cai.Lu2008}.
Matchgates were first considered in \cite{Valiant2008}.
\begin{defn}
\label{def: module}For disjoint sets $A$ and $B$, and for assignments
$x\in\{0,1\}^{A}$ and $y\in\{0,1\}^{B}$, we write $xy\in\{0,1\}^{A\cup B}$
for the assignment that agrees with $x$ on $A$, and with $y$ on
$B$. We also say that the assignment $xy$ \emph{extends} $x$.

A \emph{gate} is a signature graph $\Gamma$ containing a set $D\subseteq E(\Gamma)$
of dangling edges, all having edge-weight $1$. The \emph{signature
realized by $\Gamma$ }is the function $\Sig(\Gamma):\{0,1\}^{D}\to\mathbb{Q}$
that maps $x$ to 
\begin{equation}
\Sig(\Gamma,x)=\sum_{y\in\{0,1\}^{E(\Gamma)\setminus D}}\left(\prod_{e\in xy}w(e)\right)\left(\prod_{v\in V(\Gamma)}f_{v}(xy|_{I(v)})\right).\label{eq: module signature}
\end{equation}
A gate $\Gamma$ is a matchgate if it features only the signature
$\sigHW{=1}$.
\end{defn}
In the following, we consider the dangling edges $D$ of gates $\Gamma$
to be labelled as $1,\ldots,|D|$. This way, we can view $\Sig(\Gamma)$
as a function of type $\{0,1\}^{|D|}\to\mathbb{Q}$ instead of $\{0,1\}^{D}\to\mathbb{Q}$.
We will use gates to realize required signatures as ``gadgets''
consisting of other (usually simpler) signatures. Consider the following
example, which appeared in \cite{Valiant2008}.
\begin{example}
It can be verified that $\mathtt{EVEN}_{3}$ and $\mathtt{ODD}_{3}$
are realized by the matchgates $\Gamma_{0}$ and $\Gamma_{1}$ below,
where all vertices are assigned $\sigHW{=1}$.

\begin{center}
\includegraphics[width=0.4\textwidth]{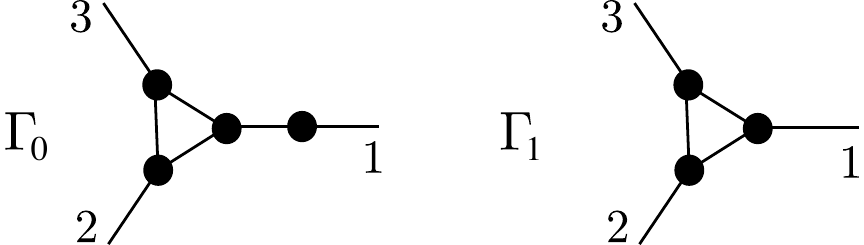}
\par\end{center}

\end{example}
Using this, we can realize the signatures $\mathtt{ODD}_{k}$ and
$\mathtt{EVEN}_{k}$ for any arity $k\geq3$ as matchgates, noted
in a similar way in \cite[Theorem 3.3]{Valiant2008}. This will be
required in Section~\ref{sec: Parity Separation}.
\begin{example}
\label{exa: module-odd}For all $k\geq3$, there exists a gate $\Gamma_{\mathtt{EVEN}}$
with $\Sig(\Gamma_{\mathtt{EVEN}})=\mathtt{EVEN}_{k}$. It consists
of vertices $v_{1},\ldots,v_{k-2}$ equipped with $\mathtt{EVEN}_{3}$,
edges $e_{1},\ldots,e_{k-3}$, and dangling edges $[k]$.

\begin{center}
\includegraphics[width=0.7\textwidth]{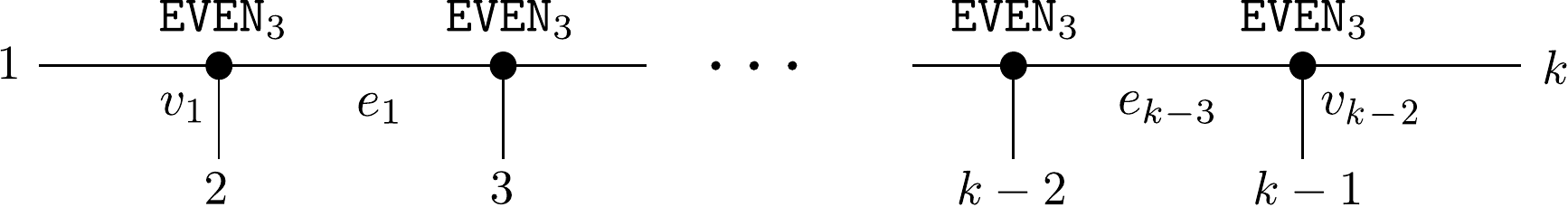}
\par\end{center}

We can likewise realize $\sigOdd_{k}$ by a gate $\Gamma_{\mathtt{ODD}}$
as above, but with $\sigOdd_{3}$ rather than $\mathtt{EVEN}_{3}$
at $v_{k-2}$.\end{example}
\begin{proof}
Let $x$ be a satisfying assignment to $\Gamma_{\mathtt{EVEN}}$.
By $\mathtt{EVEN}_{3}$ at $v_{1}$, we have $x(e_{1})=x(1)\oplus x(2)$,
where $\oplus$ denotes addition in $\mathbb{Z}/2\mathbb{Z}$. Likewise,
we have $x(e_{2})=x(e_{1})\oplus x(3)$, so we obtain inductively
that 
\begin{equation}
x(e_{k-3})=\bigoplus_{t=1}^{k-2}x(t).\label{eq: oddcalc}
\end{equation}
Then $\mathtt{EVEN}_{3}$ at $v_{k-2}$ implies that 
\begin{eqnarray*}
x(e_{k-3})\oplus x(k-1)\oplus x(k) & \underset{\eqref{eq: oddcalc}}{=} & \left(\bigoplus_{t=1}^{k-2}x(t)\right)\oplus x(k-1)\oplus x(k)\ =\ \bigoplus_{t=1}^{k}x(t)=0.
\end{eqnarray*}
The same argument applies for $\Gamma_{\mathtt{ODD}}$.
\end{proof}
In the following, we formalize the operation of \emph{inserting} a
gate $\Gamma$ into a signature graph so as to simulate a desired
signature.
\begin{lem}
\label{lem: module-contract}Let $\Omega$ be a signature graph, let
$v\in V(\Omega)$ with $D=I(v)$ and let $\Gamma$ be a gate with
dangling edges $D$. We can \emph{insert $\Gamma$ at $v$} by deleting
$v$ and keeping $D$ as dangling edges, and then placing $\Gamma$
into $\Omega$ and identifying each dangling edge $e\in D$ across
$\Gamma$ and $\Omega$. If $\Omega'$ is derived from $\Omega$ by
inserting a gate $\Gamma$ with $\Sig(\Gamma)=f_{v}$ at $v$, then
$\Holant(\Omega)=\Holant(\Omega')$. 
\end{lem}
By an argument from the author's PhD thesis \cite{Curticapean.PhD},
also used in \cite{DBLP:conf/soda/CM_Unpublished}, we can realize
every even signature $f$ by some matchgate $\Gamma=\Gamma(f)$. If
the image of $f$ is $W$, then $\Gamma$ contains $W\cup\{\pm1,\nicefrac{1}{2}\}$
as edge-weights. This yields a reduction from Holant problems to $\PerfMatch$
that we use in Sections~\ref{sub: Application sharpP} and \ref{sub: Application CeqP}.
\begin{lem}[\cite{Curticapean.PhD}]
\label{lem: Realize every Holant}Let $\Omega$ be a signature graph
on $n$ vertices and $m$ edges, with even vertex functions $\{f_{v}\}_{v\in V(\Omega)}$
that map into $W\subseteq\mathbb{Q}$. Let $s=\max_{v\in V(\Omega)}|\supp(f_{v})|$.
Then we can construct, in linear time, a graph $G$ on $\mathcal{O}\left(n+sm\right)$
vertices and edges such that $\Holant(\Omega)=\PerfMatch(G)$. The
edge-weights of $G$ are $W\cup\{\pm1,\nicefrac{1}{2}\}$.\end{lem}

\section{\label{sec: Parity Separation}The parity separation technique}

We are ready to prove Lemma~\ref{main lem: Difference technique},
our main result. The proof proceeds by establishing, with several
intermediate steps, the reduction chain 
\begin{equation}
\PerfMatch^{-1,0,1}\leq_{p}\MatchSum^{-1,0,1}\leq_{p}^{T}\PerfMatch^{0,1}.\label{eq: Parity Chain}
\end{equation}
For the first reduction in (\ref{eq: Parity Chain}), we apply a gadget
$\Gamma$ realizing the signature $\mathtt{EDGE}_{-1}$ from Lemma~\ref{lem: Removing Holant weights}
to all edges of weight $-1$.
\begin{lem}
\label{lem: Edge-gate}We have $\mathtt{EDGE}_{-1}=\Sig(\Gamma)$,
where $\Gamma$ is the gate in Figure~\ref{fig: Gamma}. In $\Gamma$,
each vertex features the signature $\mathtt{VTX}_{w}$ for the number
$w\in\{-1,0,1\}$ it is annotated with in the figure.
\end{lem}
\begin{figure}
\begin{centering}
\includegraphics[width=3cm]{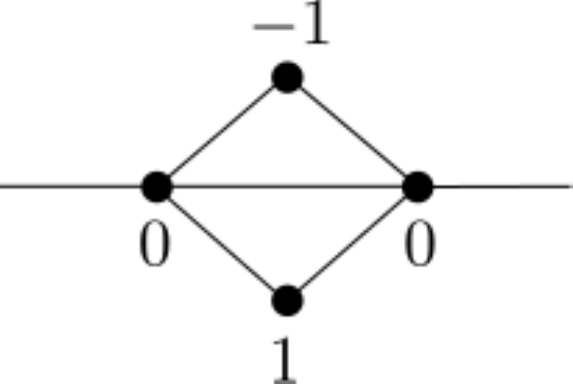}
\par\end{centering}

\caption{\label{fig: Gamma}The gate $\Gamma$.}
\end{figure}

\begin{proof}
Given an assignment $x\in\{0,1\}^{2}$ to the dangling edges of $\Gamma$,
we list the satisfying assignments $xy\in\{0,1\}^{E(\Gamma)}$ that
extend $x$ in Figure~\ref{fig: Gamma states}.
\begin{figure}
\begin{centering}
\includegraphics[width=0.7\textwidth]{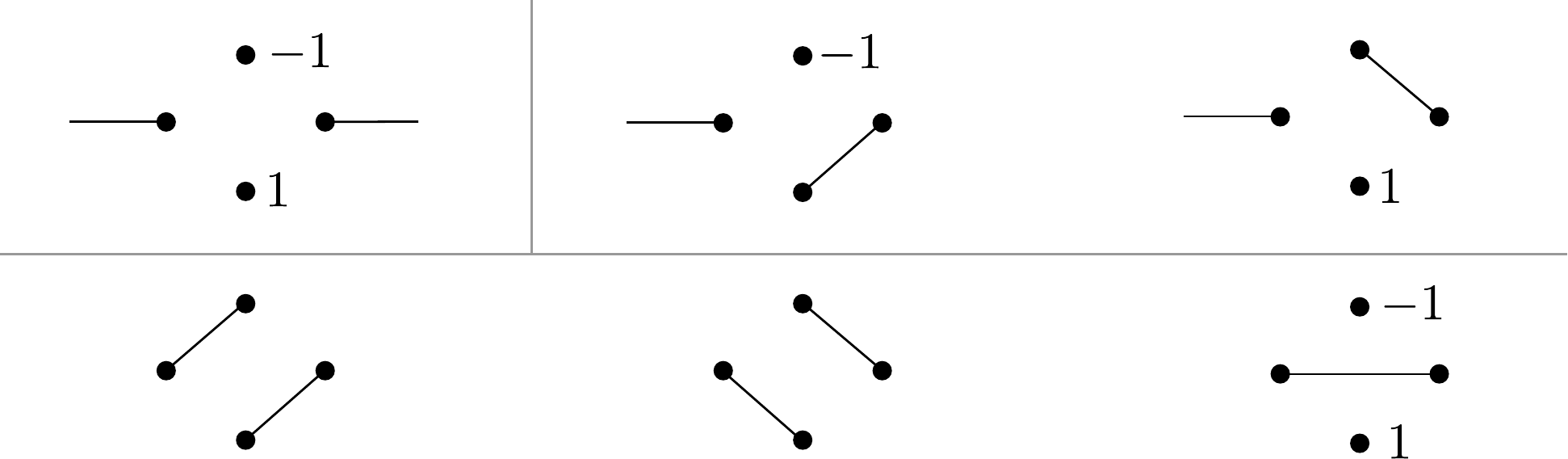}
\par\end{centering}

\caption{\label{fig: Gamma states}The satisfying assignments for $\Gamma$,
grouped by their sets of active dangling edges.}
\end{figure}
 Note that all such assignments are (not necessarily perfect) matchings.
\begin{lyxlist}{00.00.0000}
\item [{$x=11$:}] Only the empty matching can be chosen. It has weight
$-1$, thus $\Sig(\Gamma,11)=-1$.
\item [{$x=10$:}] Two matchings can be chosen, which have opposite weights,
thus $\Sig(\Gamma,10)=0$. By symmetry, the same is true for $\Sig(\Gamma,01)$.
\item [{$x=00$:}] Three matchings can be chosen, of which two have weight
$1$ and one has weight $-1$, thus $\Sig(\Gamma,00)=1$.
\end{lyxlist}
This proves the claim.
\end{proof}
This allows us to transform an instance for $\PerfMatch^{-1,0,1}$
to one for $\MatchSum^{-1,0,1}$. 
\begin{lem}
\label{lem: PerfMatch to MatchSum}Let $G$ be a graph with $n$ vertices
and $m$ edges, all of weight $\pm1$. Then we can compute a graph
$G'$ on $\mathcal{O}(n+m)$ edges, with vertices of weight $\{-1,0,1\}$,
such that $\PerfMatch(G)=\MatchSum(G')$.\end{lem}
\begin{proof}
We assume that $|V(G)|$ is even, as otherwise $\PerfMatch(G)=0$.
First, let $\Omega$ be the signature graph constructed by assigning
$\sigHW{=1}$ to all vertices of $G$, and then applying the signature
$\mathtt{EDGE}_{-1}$ as in Lemma~\ref{lem: Removing Holant weights}.
We obtain $\PerfMatch(G)=\Holant(\Omega)$. 

Then realize each occurrence of $\mathtt{EDGE}_{-1}$ by the gate
$\Gamma$ from Lemma~\ref{lem: Edge-gate}. Note that $\Gamma$ features
no edge-weights, and only the signature $\mathtt{VTX}_{w}$ for $w\in\{-1,0,1\}$.
We obtain a signature graph $\Omega'$ whose signatures are all of
the type $\mathtt{VTX}_{w}$ for $w\in\{-1,0,1\}$, and which satisfies
$\Holant(\Omega)=\Holant(\Omega')$. Note that $\sigHW{=1}=\mathtt{VTX}_{0}$,
so this indeed covers all vertices of $\Omega'$.

By Lemma~\ref{lem: Holant MatchSum}, we may equivalently consider
$\Holant(\Omega')=\MatchSum(G')$, where $G'$ is a vertex-weighted
graph obtained from $\Omega'$ as follows: Keep all vertices and edges
of $\Omega'$ intact, and if $v\in V(\Omega')$ features the signature
$\mathtt{VTX}_{w}$, for $w\in\{-1,0,1\}$, then assign the vertex
weight $w$ to $v$ in $G'$. 
\end{proof}
For the second reduction in (\ref{eq: Parity Chain}), we perform
the actual act of parity separation: We will split the vertex-weighted
graph $G'$ into an even part $G_{0}$ and an odd part $G_{1}$, both
unweighted, such that the \emph{perfect} matchings of the even (resp.
odd) part correspond bijectively to the matchings of $G'$ with an
even (resp. odd) number of unmatched vertices of weight $-1$. Since
$(-1)^{\mathit{even}}=1$ and $(-1)^{\mathit{odd}}=-1$, this clearly
implies that $\MatchSum(G)$ is the difference of $\PerfMatch(G_{0})$
and $\PerfMatch(G_{1})$.

To proceed, we first use the signatures $\sigEven$ and $\sigOdd$
from Example~\ref{exa: signatures} to obtain an alternative reformulation
of $\MatchSum^{-1,0,1}$ as the difference of two Holants.
\begin{lem}
\label{lem: MatchSum to Holant}Let $G'$ be a graph with vertex-weights
$\{-1,0,1\}$. For $a,b\in\{0,1\}$, let $\Phi_{ab}=\Phi_{ab}(G')$
be the signature graph obtained as follows:
\begin{enumerate}
\item Assign the signature $\sigHW{=1}$ to all vertices of $G'$.
\item For $x\in\{-1,0,1\}$, let $V_{x}\subseteq V(G')$ denote the set
of vertices of weight $x$ in $G'$. For $x\in\{-1,1\}$, add a vertex
$u_{x}$ connected to $V_{x}$. Assign to $u_{-1}$ the signature
$\sigEven$ if $a=0$, and assign $\sigOdd$ if $a=1$. Likewise,
assign to $u_{1}$ the signature $\sigEven$ if $b=0$, and assign
$\sigOdd$ if $b=1$.
\end{enumerate}
Then we have $\MatchSum(G')=\Holant(\Phi_{00})-\Holant(\Phi_{11})$.\end{lem}
\begin{proof}
Assume that $G'$ has an even number of vertices, so every matching
of $G'$ has an even number of unmatched vertices. For matchings $M\in\mathcal{M}[G']$,
let $w(M)=\prod_{v\in\usat(M)}w(v)$. If $w(M)\neq0$, then we have
$\usat(M)\subseteq V_{-1}\cup V_{1}$. 

For every $M\in\M[G']$, either $|\usat(M)\cap V_{-1}|$ and $|\usat(M)\cap V_{1}|$
are both even, or both are odd. In the first case, we can uniquely
extend $M$ to a satisfying assignment for $\Phi_{00}$: For every
unmatched vertex $v\in\usat(M)$ of weight $x\in\{-1,1\}$, include
the edge from $v$ to $u_{x}$. Then the signatures $\sigHW{=1}$
at vertices other than $u_{-1}$ and $u_{1}$ yield $1$, and the
signatures $\sigEven$ at $u_{-1}$ and $u_{1}$ yield $1$ as well.
Conversely, satisfying assignments to $\Phi_{00}$ can be mapped to
unique matchings of $G'$ with an even number of unmatched vertices
of weight $-1$ and $1$ each. The same correspondence can be established
between $\Phi_{11}$ and the matchings of $G'$ with an \emph{odd
}number of unmatched vertices of weight $-1$ and $1$ each.

It is clear that matchings $M$ with oddly many unmatched vertices
of weight $-1$ have $w(M)=-1$, while those with an even number satisfy
$w(M)=1$. This proves the lemma.
\end{proof}
The second reduction in (\ref{eq: Parity Chain}) follows by realizing
the signatures $\sigOdd$ and $\sigEven$ appearing in $\Phi_{00}$
and $\Phi_{11}$ via matchgates that feature neither edge- nor vertex-weights.
Note that the only other appearing signature $\sigHW{=1}$ is trivially
realized by such a matchgate.
\begin{proof}[Proof of Lemma~\ref{main lem: Difference technique}]
Follows from Lemma~\ref{lem: PerfMatch to MatchSum} (to reduce
$\PerfMatch^{-1,0,1}$ to $\MatchSum^{-1,0,1}$) with Lemma~\ref{lem: MatchSum to Holant}
(to reformulate $\MatchSum^{-1,0,1}$ as a Holant problem) and Example~\ref{exa: module-odd}
(to realize the $\sigOdd$ and $\sigEven$ signatures occurring in
the Holant problem by unweighted matchgates).
\end{proof}
We remark that the proof could also be expressed in the framework
of combined signatures introduced in~\cite{Curtican.Xia2015}. A
presentation along these lines can be found in \cite{Curticapean.PhD}.

\section{\label{sec: Applications}Parity separation in action}

In the final section of this paper, we cover the three applications
of parity separation that we discussed in the introduction.

\subsection{\label{sub: Application sharpP}Completeness for $\protect\sharpP$}

We can easily show the $\sharpP$-completeness of $\PerfMatch^{0,1}$
via parity separation. To this end, we first express $\sharpSAT$
as a Holant problem on even signature graphs, as seen in Lemma~\ref{lem: Holant SAT}.
Together with Lemma~\ref{lem: Realize every Holant}, this yields
$\sharpSAT\leqPoly\PerfMatch^{B}$ with $B=\{-1,0,\nicefrac{1}{2},1\}$.
We use Lemma~\ref{lem: Fractional weights in PerfMatch} to remove
the edge-weight $\nicefrac{1}{2}$, and finally remove the weight
$-1$ by parity separation as in Lemma~\ref{main lem: Difference technique}.
This yields the following lemma.
\begin{lem}
\label{lem: SAT =00003D difference}Let $\varphi$ be a $3$-CNF formula
with $n$ variables and $m$ clauses. Then we can compute a number
$T\in\mathbb{N}$ and construct two unweighted graphs $G_{1}$ and
$G_{2}$ on $\mathcal{O}(n+m)$ vertices and edges, all in time $\mathcal{O}(n+m)$,
such that $2^{T}\cdot\sharpSAT(\varphi)=\PerfMatch(G_{1})-\PerfMatch(G_{2}).$
\end{lem}
This readily implies Theorem~\ref{thm: sharpP-completeness}, the
desired $\sharpP$-completeness result.

\subsection{\label{sub: Application CeqP}Completeness for $\protect\CeqP$}

For our next application, we apply the parity separation technique
to prove Theorem~\ref{main thm: CeqP}. That is, we prove $\CeqP$-completeness
of the problem $\PerfMatch_{=}^{0,1}$ that asks, given two unweighted
graphs $G_{1}$ and $G_{2}$, whether their numbers of perfect matchings
agree. We call graphs satisfying this property \emph{equipollent graphs}
and will likewise speak of \emph{equipollent formulas }if their numbers
of satisfying assignments agree.
\begin{proof}[Proof of Theorem~\ref{main thm: CeqP}]
The problem $\PerfMatch_{=}^{0,1}$ is clearly contained in $\CeqP$.
For the hardness part, we reduce from the $\CeqP$-complete problem
$\sharpSAT_{=}$ that asks, given 3-CNF formulas $\varphi$ and $\varphi'$,
to determine whether they are equipollent. To this end, we construct
unweighted graphs $G$ and $G'$ that are equipollent if and only
if $\varphi$ and $\varphi'$ are.

Assume that $\varphi$ and $\varphi'$ are defined on the same set
of variables $x_{1},\ldots,x_{n}$ and feature the same number $m$
of clauses. This can be achieved by renaming variables, and by adding
dummy variables and clauses. If, say, $\varphi$ has less variables
than $\varphi'$, then we can add dummy variables to $\varphi'$,
together with clauses that ensure that every dummy variable has the
same assignment as $x_{1}$. We can also duplicate clauses.

Let $C_{1},\ldots,C_{m}$ and $C'_{1},\ldots,C'_{m}$ denote the clauses
in $\varphi$ and $\varphi'$, respectively. We introduce a \emph{selector}
variable $x^{*}$ and define a formula $\psi$ on the variable set
$\mathcal{X}=\{x^{*},x_{1},\ldots,x_{n}\}$, which has clauses $D_{1},\ldots,D_{m}$
and $D'_{1},\ldots,D'_{m}$, where $D_{i}:=(x^{*}\vee C_{i})$ and
$D_{i}':=(\neg x^{*}\vee C'_{i})$ for $i\in[m]$. If $a(x^{*})=0$
holds in an assignment $a\in\{0,1\}^{\mathcal{X}}$, then all clauses
$D'_{1},\ldots,D'_{m}$ are satisfied by $\neg x^{*}$, but in order
for $a$ to satisfy $\psi$, the clauses $D_{1},\ldots,D_{m}$ have
to be satisfied by $x_{1},\ldots,x_{n}$. In other words, if $a$
satisfies $\psi$ and $a(x^{*})=0$, then the restriction of $a$
to $x_{1},\ldots,x_{n}$ satisfies $\varphi$, and if $a$ satisfies
$\psi$ and $a(x^{*})=1$, then the restriction of $a$ to $x_{1},\ldots,x_{n}$
satisfies $\varphi'$. 

Hence, we can define the following quantity 
\[
S:=\sum_{a\in\{0,1\}^{\mathcal{X}}}(-1)^{a(x^{*})}\cdot[\psi\mbox{ satisfied by }a]
\]
and we observe that $S=\sharpSAT(\varphi)-\sharpSAT(\varphi')$. It
is clear that $S=0$ if and only if $\varphi$ and $\varphi'$ are
equipollent. As in Lemma~\ref{lem: Holant SAT}, we then express
$S=\Holant(\Omega)$ for a signature graph $\Omega=\Omega(\psi)$,
with one modification: At the vertex $v^{*}$ corresponding to the
variable $x^{*}$, we replace the signature $\EQ$ by a modified signature
\[
\mathtt{EQ}_{-}:\quad y\mapsto\begin{cases}
-1 & \mbox{if }y=1\ldots1,\\
1 & \mbox{if }y=0\ldots0,\\
0 & \mbox{otherwise}.
\end{cases}
\]

We realize $\Omega$ via Lemma~\ref{lem: Realize every Holant} to
obtain a graph $G$, simulate the edge-weight $\nicefrac{1}{2}$ via
Lemma~\ref{lem: Fractional weights in PerfMatch}, and obtain an
edge-weighted graph $H$ with weights $\pm1$ together with a number
$T\in\mathbb{N}$ such that 
\begin{equation}
S=\Holant(\Omega)=2^{-T}\cdot\PerfMatch(H).\label{Q:Unweighted: eq: Holant =00003D WeightedPM}
\end{equation}
Using Lemma~\ref{main lem: Difference technique}, we then obtain
unweighted graphs $G$ and $G'$ such that 
\begin{equation}
\PerfMatch(H)=\PerfMatch(G)-\PerfMatch(G').\label{Q:Unweighted: eq: PM =00003D PM01}
\end{equation}
It is clear that $G$ and $G'$ are equipollent iff $S=0$, which
in turn holds iff $\varphi$ and $\varphi'$ are equipollent.
\end{proof}

\subsection{\label{sub: Application ETH}Tight lower bounds under $\protect\sharpETH$}

By the exponential-time hypothesis $\sharpETH$, there is no $2^{o(n)}$
time algorithm for counting satisfying assignments to $3$-CNF formulas
$\varphi$ with $n$ variables. Applying the counting version of the
so-called sparsification lemma, shown in \cite{Dell.Husfeldt2014},
we may additionally assume that $\varphi$ features $m=\mathcal{O}(n)$
clauses. Then Lemma~\ref{lem: SAT =00003D difference} clearly implies
the lower bound for $\PerfMatch^{0,1}$ claimed in Theorem~\ref{main thm: Bounds under ETH}.

Concerning $\PerfMatch_{=}^{0,1}$, it is even easier to prove lower
bounds under $\ETH$ than to prove its $\CeqP$-completeness, as we
may (i) reduce from $\SAT$ rather than $\SAT_{=}$, and (ii) use
the more permissive notion of Turing (rather than many-one) reductions:
With Lemma~\ref{lem: SAT =00003D difference}, we can construct unweighted
graphs $G_{1}$ and $G_{2}$ on $\mathcal{O}(m)$ vertices and edges
that are equipollent iff $\varphi$ is \emph{unsatisfiable}, thus
a $2^{\mathcal{O}(m)}$ time algorithm would contradict $\ETH$. This
proves Theorem~\ref{main thm: Bounds under ETH}.

\section{Conclusion and future work}

We have added a new method to the known techniques (modular arithmetic
and polynomial interpolation) for removing the edge-weight $-1$ from
$\PerfMatch^{-1,0,1}$. This method is based on matchgates and the
rather trivial observation that $(-1)^{\mathit{even}}=1$ and $(-1)^{\mathit{odd}}=-1$.
We obtained non-trivial applications that could not be obtained via
the previously known techniques.

Our work leaves several questions open for further investigations.
For instance, we could not find a way to show $\sharpP$-completeness
of $\PerfMatch^{0,1}$ on \emph{bipartite }graphs by following the
outline of parity separation. Does this admit a complexity-theoretic
explanation or are we to blame? On another note, can parity separation
also be adapted to, say, proving $\CeqP$-completeness for other ``equality-testing''
versions of counting problems?

\bibliographystyle{plain}
\bibliography{References}

\end{document}